\documentclass[12pt,reqno]{amsart}

\newcommand\version{October 1, 2013}

\usepackage{amsmath,amsfonts,amsthm,amssymb,amsxtra}

\newtheorem{theorem}{Theorem}
\newtheorem{proposition}[theorem]{Proposition}
\newtheorem{lemma}[theorem]{Lemma}

\theoremstyle{definition}

\theoremstyle{remark}

\newtheorem{remark}[theorem]{Remark}

\newcommand{\C}{\mathbb{C}}

\renewcommand{\epsilon}{\varepsilon}

\renewcommand{\phi}{\varphi}

\DeclareMathOperator{\tr}{Tr}

\begin{document}

\title[Monotonicity of a relative R\'enyi entropy --- \version]{Monotonicity of a relative R\'enyi entropy}

\author{Rupert L. Frank}
\address{Rupert L. Frank, Mathematics 253-37, Caltech, Pasadena, CA 91125, USA}
\email{rlfrank@caltech.edu}

\author{Elliott H. Lieb}
\address{Elliott H. Lieb, Departments of Mathematics and Physics, Princeton University, Washington Road, Princeton, NJ 08544, USA}
\email{lieb@princeton.edu}

\thanks{\copyright\, 2013 by the authors. This paper may be
reproduced, in its entirety, for non-commercial purposes.\\
U.S.~National Science Foundation grants PHY‐1347399 (R.F.), PHY-0965859 and PHY-1265118 (E.L.) and the Simons Foundation grant 230207 (E.L.) are acknowledged.}

\thanks{\version}

\begin{abstract}
We show that a recent definition of relative R\'enyi entropy is monotone under completely positive, trace preserving maps. This proves a recent conjecture of M\"uller--Lennert et al.
\end{abstract}

\maketitle


Recently, M\"uller--Lennert et al. \cite{MuLe} and Wilde et al. \cite{Wi} modified the traditional notion of relative R\'enyi entropy and showed that their new definition has several desirable properties of a relative entropy. One of the fundamental properties of a relative entropy, namely monotonicity under completely positive, trace preserving maps (quantum operations) was shown only in a limited range of parameters and conjectured for a larger range. Our goal here is to prove this conjecture.

More precisely, the definition of the \emph{quantum R\'enyi divergence} \cite{MuLe} or \emph{sandwiched R\'enyi entropy} \cite{Wi} is
$$
D_\alpha(\rho\|\sigma)=
\begin{cases}
(\alpha-1)^{-1} \log\left( \left(\tr\rho\right)^{-1} \tr\left(\sigma^{(1-\alpha)/(2\alpha)} \rho \sigma^{(1-\alpha)/(2\alpha)} \right)^\alpha \right)  & \text{if}\ \alpha\in(0,1)\cup(1,\infty) \,, \\
\left(\tr\rho\right)^{-1} \tr\rho\left(\log\rho - \log\sigma\right) & \text{if}\ \alpha=1 \,, \\
\log \left\|\sigma^{-1/2}\rho\sigma^{-1/2}\right\|_\infty & \text{if}\ \alpha=\infty
\end{cases}
$$
for non-negative operators $\rho,\sigma$. Here, for $\alpha\geq 1$, we define $\tr\left( \sigma^{(1-\alpha)/\alpha} \rho \sigma^{(1-\alpha)/\alpha} \right)^\alpha=\infty$ if the kernel of $\sigma$ is not contained in the kernel of $\rho$. The factor $(\tr\rho)^{-1}$ is inessential and could be dropped, but we keep it in order to be consistent with \cite{MuLe}. After a first version of our paper appeared (arXiv:1306.5358) we were made aware of the fact that $D_\alpha(\rho\|\sigma)$ is a special case of a two-parameter family of relative entropies introduced earlier in \cite{JOPP}.

Note that $D_\alpha(\rho\|\sigma)$ is the relative von Neumann entropy for $\alpha=1$, the relative max-entropy for $\alpha=\infty$ and closely related to the fidelity $\tr\left(\sigma^{1/2}\rho\sigma^{1/2}\right)^{1/2}$ for $\alpha=1/2$. In \cite{MuLe} it is shown that $D_\alpha(\rho\|\sigma)$ depends continuously on $\alpha$, in particular, at $\alpha=1$ and $\alpha=\infty$.

The definition of $D_\alpha(\rho\|\sigma)$ should be compared with the traditional relative R\'enyi entropy (see e.g. \cite{MoHi}),
$$
D'_\alpha(\rho\|\sigma) = (\alpha-1)^{-1} \log\left( \left(\tr\rho\right)^{-1} \tr \sigma^{1-\alpha} \rho^\alpha \right)  \quad \text{if}\ \alpha\in(0,1)\cup(1,\infty) \,.
$$
Note that by the Lieb--Thirring trace inequality \cite{LiTh}
$$
D_\alpha(\rho\|\sigma) \leq D'_\alpha(\rho\|\sigma)
\qquad
\text{for}\ \alpha >1 \,.
$$

Our main results in this paper are the following two theorems.

\begin{theorem}[Monotonicity]\label{main}
Let $1/2\leq\alpha\leq\infty$ and let $\rho,\sigma\geq 0$. Then for any completely 
positive, trace preserving map $\mathcal E$,
$$
D_\alpha(\rho\|\sigma) \geq D_\alpha(\mathcal E(\rho)\|\mathcal E(\sigma)) \,.
$$
\end{theorem}

\begin{theorem}[Joint convexity]\label{main2}
Let $1/2\leq\alpha\leq 1$. Then $D_\alpha(\rho\|\sigma)$ is jointly convex on 
pairs $(\rho,\sigma)$ of non-negative operators with $\tr\rho= t$ for any fixed $t>0$.
\end{theorem}

For the relative von Neumann entropy ($\alpha=1$) both theorems are due to Lindblad \cite{Lb}, 
whose proof is based on Lieb's concavity theorem \cite{Li}. Theorem \ref{main} for $\alpha\in (1,2]$ 
is due to \cite{MuLe} and \cite{Wi}. In a preprint of \cite{MuLe} its validity was conjectured for
 all values $\alpha\geq 1/2$. Shortly after the first version of our paper appeared (arXiv:1306.5358v1) which proved this conjecture for all $\alpha\geq 1/2$, Beigi independently posted (arXiv:1306.5920) an alternative
proof  of Theorem \ref{main} in the range $\alpha\in(1,\infty)$.

Just as in Lindblad's monotonicity proof for $\alpha=1$, we will deduce Theorem \ref{main} 
for $\alpha>1$ from Lieb's concavity theorem \cite{Li}. The proof for $1/2\leq\alpha<1$ uses 
a close relative of this theorem, namely, Ando's convexity theorem \cite{An}. These theorems 
enter in the proof of Proposition \ref{jointcon} below.

\bigskip

Let us turn to the proofs of the theorems. Both of them are based on the following proposition.

\begin{proposition}
\label{jointcon}
The following map on pairs of non-negative operators
$$
(\rho,\sigma)\mapsto\tr\left( \sigma^{(1-\alpha)/(2\alpha)} \rho \sigma^{(1-\alpha)/(2\alpha)} \right)^\alpha
$$
is jointly concave for $1/2\leq\alpha<1$ and jointly convex for $\alpha>1$.
\end{proposition}

We note that this proposition implies that $\exp( (\alpha-1) D_\alpha(\rho\|\sigma))$ is jointly concave for $1/2\leq\alpha<1$ and jointly convex for $\alpha>1$ on pairs $(\rho,\sigma)$ of non-negative operators with $\tr\rho= t$ for any fixed $t>0$. Since $x\mapsto x^{1/(\alpha-1)}$ is increasing and convex for $1<\alpha\leq 2$, we deduce that $\exp( D_\alpha(\rho\|\sigma))$ is jointly convex for $1<\alpha\leq 2$ on pairs $(\rho,\sigma)$ of non-negative operators with $\tr\rho= t$ for any fixed $t>0$. This fact is also proved in \cite{MuLe} and \cite{Wi}.

The argument to derive Theorem \ref{main} from Proposition \ref{jointcon} is well known, but we include it for the sake of completeness. The fact that joint convexity implies monotonicity appears in \cite{Lb}, but here we also use ideas from \cite{Uh}.

\begin{proof}
[Proof of Theorem \ref{main} given Proposition \ref{jointcon}]
We prove the assertion for $\alpha\in[1/2,1)\cup(1,\infty)$. The remaining two cases follow by continuity in $\alpha$. By a limiting argument we may assume that the underlying Hilbert space is $\C^N$ for some finite $N$. If $\mathcal E$ is a completely positive, trace preserving map then by the Stinespring representation theorem \cite{St} there is an integer $N'\leq N^2$, a density matrix $\tau$ on $\C^{N'}$ (which can be chosen to be pure) and a unitary $U$ on $\C^N\otimes\C^{N'}$ such that
$$
\mathcal E(\gamma) = \tr_2 U \left(\gamma\otimes \tau \right) U^* \,.
$$
Thus, if $du$ denotes normalized Haar measure on all unitaries on $\C^{N'}$, then
\begin{equation}
\label{eq:uhlmann}
\mathcal E(\gamma) \otimes (N')^{-1} 1_{\C^{N'}} = \int (1\otimes u) U \left(\gamma\otimes \tau \right) U^* (1\otimes u^*) \,du \,.
\end{equation}
By the tensor property of $D_{\alpha}(\cdot\|\cdot)$,
\begin{equation} \label{two}
D_\alpha(\mathcal E(\rho)\|\mathcal E(\sigma))
= D_\alpha(\mathcal E(\rho)\otimes (N')^{-1} 1_{\C^{N'}} \|\, \mathcal E(\sigma)\otimes (N')^{-1} 1_{\C^{N'}}) \,.
\end{equation}
By \eqref{eq:uhlmann} and Proposition \ref{jointcon} the double, normalized $u$ integral in (\ref{two})
is bounded from 
below (if $1/2\leq\alpha<1$) or above (if $\alpha>1$) by a single integral:
\begin{align*}
& \int D_\alpha( (1\otimes u) U \left(\rho\otimes \tau \right) U^* (1\otimes u^*) \|\, (1\otimes u) U \left(\sigma\otimes \tau \right) U^* (1\otimes u^*) ) \,du \\
& \quad = \int D_\alpha( \rho\otimes \tau \| \sigma\otimes \tau ) \,du \\
& \quad = D_\alpha( \rho\otimes \tau \| \sigma\otimes \tau ) \\
& \quad = D_\alpha( \rho \| \sigma ) \,.
\end{align*}
Here, we used the unitary invariance of $D_{\alpha}(\cdot\|\cdot)$, the normalization of the Haar measure and the tensor property of $D_{\alpha}(\cdot\|\cdot)$.

Dividing the inequality we have obtained by $\tr\mathcal E(\rho)=\tr\rho$, taking logarithms and multiplying by $\alpha-1$ we obtain the monotonicity stated in the theorem.
\end{proof}

\begin{proof}
[Proof of Theorem \ref{main2} given Proposition \ref{jointcon}]
This follows immediately from Proposition \ref{jointcon} together with the fact that $x\mapsto\log x$ is increasing and concave.
\end{proof}

Thus, we have reduced the proofs of Theorems \ref{main} and \ref{main2} to the proof of Proposition \ref{jointcon}. The latter, in turn, is based on two ingredients. The first one is a representation formula for $\tr\left( \sigma^{(1-\alpha)/(2\alpha)} \rho \sigma^{(1-\alpha)/(2\alpha)} \right)^\alpha$.

\begin{lemma}\label{repr}
Let $\rho,\sigma\geq 0$ be operators. Then, if $\alpha>1$,
$$
\tr\left( \sigma^{(1-\alpha)/(2\alpha)} \rho \sigma^{(1-\alpha)/(2\alpha)} \right)^\alpha
= \sup_{H\geq 0} \left( \alpha \tr H\rho - (\alpha-1) \tr\left( H^{1/2} \sigma^{(\alpha-1)/\alpha} H^{1/2} \right)^{\alpha/(\alpha-1)} \right) \,.
$$
The same equality holds for $0<\alpha<1$, provided $\sup$ is replaced by $\inf$.
\end{lemma}

The second ingredient in the proof of Proposition \ref{jointcon} is a concavity result for $\tr \left( B^* A^p B\right)^{1/p}$.

\begin{lemma}\label{conc}
For a fixed operator $B$, the map on positive operators
$$
A\mapsto \tr \left( B^* A^p B\right)^{1/p}
$$
is concave for $-1\leq p\leq 1$, $p\neq 0$.
\end{lemma}

The case $0< p\leq 1$ in this lemma is due to Epstein \cite{Ep}, with an alternative proof 
due to Carlen--Lieb \cite{CaLi} based on the Lieb concavity theorem  \cite{Li}.
Legendre transforms, similar to Lemma \ref{repr}, are also
used in \cite{CaLi}.

The remaining case $-1\leq p<0$ can be proved similarly, using Ando's convexity
theorem \cite{An}, as 
in \cite{CaLi}.  (For 
an introduction to both theorems we refer to \cite{Ca}.)
While this case could easily have been included in  \cite{CaLi}, it was not, and
for the benefit of the reader we explain the argument below.   Alternatively,
 one could probably follow Bekjan's adaption \cite{Bj} of Epstein's proof to establish 
the  $-1\leq p<0$ case.
\begin{proof}
[Proof of Proposition \ref{jointcon} given Lemmas \ref{repr} and \ref{conc}]
Lemma \ref{conc} implies that
$$
\sigma\mapsto (1-\alpha) \tr\left( H^{1/2} \sigma^{(\alpha-1)/\alpha} H^{1/2} 
\right)^{\alpha/(\alpha-1)}
$$
is concave for $1/2\leq\alpha<1$ and convex for $\alpha>1$. The claim of the proposition 
now follows from the representation formula in Lemma \ref{repr}.
\end{proof}

It remains to prove the lemmas.

\begin{proof}
[Proof of Lemma \ref{repr}]
Let $\alpha>1$ and abbreviate $\beta =(\alpha-1)/(2\alpha)$. Since $H^{1/2} \sigma^{2\beta} H^{1/2}$ and $\sigma^{\beta} H \sigma^{\beta}$ have the same non-zero eigenvalues, the right side of the lemma is the same as
$$
\sup_{H\geq 0} \left( \alpha \tr H\rho - (\alpha-1) \tr\left( \sigma^{\beta} H \sigma^{\beta} \right)^{1/(2\beta)} \right) \,.
$$
Let us show that the supremum is given by $\tr\left( \sigma^{-\beta} \rho \sigma^{-\beta} \right)^\alpha$. To prove this, we may assume (by continuity) that $\sigma$ is positive and we observe that the supremum is attained (at least if the underlying Hilbert space is finite-dimensional, which we may assume again by an approximation argument). The Euler--Lagrange equation for the optimal $\hat H$ reads
$$
\alpha \rho - \alpha \sigma^{\beta} \left( \sigma^{\beta} \hat H \sigma^{\beta} \right)^{1/(\alpha-1)} \sigma^{\beta} = 0\,,
$$
that is,
$$
\hat H = \sigma^{-\beta} \left( \sigma^{-\beta} \rho \sigma^{-\beta} \right)^{\alpha-1} \sigma^{-\beta} \,. 
$$
By inserting  this into the expression we wish to maximize, we obtain $\tr\left( \sigma^{-\beta} \rho \sigma^{-\beta} \right)^\alpha$, as claimed. The proof for $0<\alpha<1$ is similar.
\end{proof}

We are grateful to the referee for suggesting the following alternative proof of Lemma~\ref{repr} for $\alpha>1$. Recall that for positive operators $X$ and $Y$ and $1<p,q<\infty$ with $1/p+1/q=1$ one has
$$
\tr XY \leq \frac{1}{p} \tr X^p + \frac{1}{q} \tr Y^q \,,
$$
with equality if $X^p=Y^q$. This implies the statement of the lemma, if we set $X= \sigma^{-\beta}\rho \sigma^{-\beta}$, $Y=\sigma^\beta H \sigma^\beta$ and $p=\alpha$, $q=\alpha/(\alpha-1)$.

\begin{proof}
[Proof of Lemma \ref{conc}]
As we have already mentioned, the result for $0<p\leq 1$ is known \cite{Ep,CaLi}. Therefore, we only give the proof for $-1\leq p<0$ and for this we adapt the argument of \cite{CaLi}. We note that
$$
p \tr\left( B^* A^p B\right)^{1/p} = \inf_{X\geq 0} \left( \tr A^{p/2} B X^{1-p} B^* A^{p/2} - (1-p) \tr X \right) \,.
$$
(The proof is similar to that of Lemma \ref{repr}.) If we can prove that
\begin{equation}
\label{eq:conv}
(A,X)\mapsto \tr A^{p/2} B X^{1-p} B^* A^{p/2}
\end{equation}
is jointly convex on pairs of non-negative operators, then $p \tr\left( B^* A^{p} B\right)^{1/p}$ 
as an infimum over \emph{jointly} convex functions is convex, (see \cite[Lemma 2.3]{CaLi}) which implies the lemma.

To prove that \eqref{eq:conv} is jointly convex, we write, as in \cite{Li},
$$
\tr A^{p/2} B X^{1-p} B^* A^{p/2} = \tr Z^{p} K^* Z^{1-p} K \,,
$$
where
$$
K = \begin{pmatrix}
0 & 0 \\ B^* & 0
\end{pmatrix} \,,
\qquad
Z = \begin{pmatrix}
A & 0 \\ 0 & X
\end{pmatrix} \,. 
$$
We can consider $K$, which is an operator in $\mathcal H \oplus \mathcal H$, as a vector in $\left( \mathcal H \oplus \mathcal H \right) \otimes \left( \mathcal H \oplus \mathcal H \right)$ and write $\tilde K$. Thus,
$$
\tr Z^{p} K^* Z^{1-p} K = \langle \tilde K, Z^{p}\otimes Z^{1-p} \tilde K \rangle \,.
$$
By Ando's convexity theorem \cite{An}, the right side is a convex function of $Z$. This is equivalent to \eqref{eq:conv} being jointly convex, as we set out to prove.
\end{proof}

\begin{remark}
More generally, for a fixed operator $B$, $A\mapsto\tr\left( B^* A^p B\right)^{q/p}$ is concave on non-negative operators for $0< |p|\leq q\leq 1$. The case $p> 0$ is due to Carlen--Lieb \cite{CaLi} and the case $p<0$ follows from similar arguments. More precisely, we can write
$$
r \tr\left( B^* A^{p} B\right)^{q/p} = \inf_{X\geq 0} \left( \tr A^{p/2} B X^{1-r} B^* A^{p/2} - (1-r) \tr X \right)
$$
with the notation $r=p/q<0$. Since
$$
\tr A^{p/2} B X^{1-r} B^* A^{p/2} = \tr Z^p K^* Z^{1-r} K
$$
with $Z$ and $K$ as in the previous proof, the more general assertion again follows from Ando's convexity theorem \cite{An}.
\end{remark}

\subsection*{Acknowledgements}
We thank E. Carlen, V. Jaksic, C.-A. Pillet and A. Vershynina for valuable comments on a first draft of this paper. We are grateful to the referee for various suggestions that helped to improve this paper.


\bibliographystyle{amsalpha}

\end{document}